\documentclass[sn-mathphys-num]{sn-jnl}% Math and Physical Sciences Numbered Reference Style 
\usepackage{graphicx}%
\usepackage{multirow}%
\usepackage{amsmath,amssymb,amsfonts}%
\usepackage{amsthm}%
\usepackage[title]{appendix}%
\usepackage{xcolor}%
\usepackage{textcomp}%
\usepackage{manyfoot}%
\usepackage{booktabs}%
\usepackage{algorithm}%
\usepackage{algorithmicx}%
\usepackage{algpseudocode}%
\usepackage{listings}%
\newtheorem{thm}{Theorem}[section]

\newtheorem*{thm*}{Theorem}
\newtheorem{cor}{Corollary}[section]

\newtheorem{lem}{Lemma}[section]

\newtheorem{prop}{Proposition}[section]

\newtheorem*{prop*}{Proposition}
\theoremstyle{definition}
\newtheorem{defn}{Definition}[section]

\newtheorem{ass}{Assumption}[section]

\theoremstyle{remark}
\newtheorem{rem}{Remark}[section]

\numberwithin{equation}{section}
\def\tr{\mathrm{Tr}}
\def\gR{\mathfrak R}
\def\gG{\mathfrak G}

\def\sp{\vspace{0.2cm}}

\newcommand{\modifica}[1]{#1}               % MODIFICA INVISIBILE

\begin{document}

\title[Quantum concentration inequalities and equivalence of the thermodynamical ensembles: an optimal mass transport approach]{Quantum concentration inequalities and equivalence of the thermodynamical ensembles: an optimal mass transport approach}

\author[1]{{Giacomo De Palma}} \email{giacomo.depalma@unibo.it}\equalcont{These authors contributed equally to this work.}
\author*[1,2]{{Davide Pastorello}} \email{davide.pastorello3@unibo.it}\equalcont{These authors contributed equally to this work.}
\affil[1]{University of Bologna, Department of Mathematics, Piazza di Porta San Donato 5, 40126 Bologna, Italy}
\affil[2]{TIFPA-INFN, via Sommarive 14, 38123 Povo (Trento), Italy}

%\author*[1,2]{\fnm{First} \sur{Author}}\email{iauthor@gmail.com}



%\affil*[1]{\orgdiv{Department}, \orgname{Organization}, \orgaddress{\street{Street}, \city{City}, \postcode{100190}, \state{State}, \country{Country}}}

%\affil[2]{\orgdiv{Department}, \orgname{Organization}, \orgaddress{\street{Street}, \city{City}, \postcode{10587}, \state{State}, \country{Country}}}

%\affil[3]{\orgdiv{Department}, \orgname{Organization}, \orgaddress{\street{Street}, \city{City}, \postcode{610101}, \state{State}, \country{Country}}}


\abstract{ We prove new concentration inequalities for quantum spin systems which apply to any local observable measured on any product state or on any state with exponentially decaying correlations. Our results do not require the spins to be arranged in a regular lattice, and cover the case of observables that contain terms acting on spins at arbitrary distance. Moreover, we introduce a local $W_1$ distance, which quantifies the distinguishability of two states with respect to local observables. We prove a transportation-cost inequality stating that the local $W_1$ distance between a generic state and a state with exponentially decaying correlations is upper bounded by a function of their relative entropy. Finally, we apply such inequality to prove the equivalence between the canonical and microcanonical ensembles of quantum statistical mechanics and the weak eigenstate thermalization hypothesis for the Hamiltonians whose Gibbs states have exponentially decaying correlations.}

\keywords{Quantum spin systems, concentration inequalities, quantum Wasserstein distance, canonical and microcanonical ensembles}



\maketitle

\section{Introduction}

Concentration inequalities provide an upper bound to the probability of the deviations of a random variable from its mean and have become a fundamental tool in mathematics, physics and computer science \cite{boucheron2013concentration}.
Concentration inequalities have been recently proved in the setting of local observables mesured on quantum spin systems with weak correlations and have explained important physical effects.
Local observables have the form
\begin{equation}\label{eq:local}
H = \sum_{\Lambda\subseteq[n]} h_\Lambda\,,
\end{equation}
where each $\Lambda$ is a subset of the set of all the spins and each $h_\Lambda$ only acts on the spins in the corresponding $\Lambda$.
The topic has originally been investigated for product states \cite{goderis1989central,hartmann2004existence,Kuwahara_2016, abrahamsen2020short,de2021quantum} and later for high-temperature Gibbs states \cite{de2022quantum,KUWAHARA2020gaussian}, time-evolved product states \cite{PRXQuantum.4.020340}, and states of spin lattices whose correlations decay exponentially with the distance \cite{anshu2016concentration}.
Concentration inequalities for Gibbs states have been the key tool to prove the equivalence between the canonical and the microcanonical ensembles of quantum statistical mechanics and the weak version of the eigenstate thermalization hypothesis \cite{brandao_equivalence_2015,tasaki2018local,alhambra2023quantum,kuwahara2020eigenstate,KUWAHARA2020gaussian,de2022quantum}.

The first part of the present paper is devoted to proving new concentration inequalities for product states and states of a quantum spin lattice with correlations that decay exponentially with the distance.
Such concentration inequalities are based on a local norm for observables, which is given by the maximum over the spins of twice the sum of the operator norms of the addends in \eqref{eq:local} that act on that spin (\autoref{defn:local}).

The best available concentration inequality for product states is \cite[Theorem 5.3]{anshu2016concentration}, and states that the concentration of local observables is exponential in the number of spins whenever each region $\Lambda$ in the sum \eqref{eq:local} intersects only $O(1)$ other regions.
Our \autoref{thm:concentrationproduct} extends this result to the setting where each region $\Lambda$ contains $O(1)$ spins and the local norm of the observable is $O(1)$.
We stress that our result covers also the case of observables with all-to-all interactions, \emph{e.g.}, where the sum in \eqref{eq:local} runs over all the couples of spins, provided that the sum of the operator norms of the addends that act on any given spin is $O(1)$.

The best available concentration inequalities for the states of a quantum spin lattice whose correlations decay exponentially with the distance are \cite[Theorem 4.2]{anshu2016concentration} and \cite[Eq. (S.17)]{kuwahara2020eigenstate}, and state that the probability of deviations from the average decays as $\exp\left(-\Omega\left(n^\frac{1}{2d+1}\right)\right)$, where $n$ is the number of spins and $d>0$ depends on the geometry of the lattice, whenever the diameter of all the regions in the sum \eqref{eq:local} is $O(1)$.
Our \autoref{thm:concentration2} extends these results to the setting where each region $\Lambda$ contains $O(1)$ spins and the local norm of the observable is $O(1)$.
Furthermore, while \cite[Theorem 4.2]{anshu2016concentration} and \cite[Eq. (S.17)]{kuwahara2020eigenstate} require a regular lattice, our results do not require any regular structure nor the notion of neighboring spins, but only require a distance on the set of the spins such that the size of each metric ball is $O(r^d)$, where $r$ is the radius (\autoref{ass:ball}).

The second part of the paper is devoted to the equivalence between the ensembles of quantum statistical mechanics.
It has long been known that a concentration inequality for an observable measured on a Gibbs state implies an upper bound in the difference between the expectation value of that observable on the Gibbs and on any microcanonical states with sufficiently close average energy \cite{brandao_equivalence_2015,tasaki2018local,alhambra2023quantum,kuwahara2020eigenstate,KUWAHARA2020gaussian,de2022quantum}, thus bounding the distinguishability between the Gibbs and the microcanonical states and proving the equivalence between the canonical and microcanonical ensembles for the observables that exhibit concentration.
Such link between concentration and equivalence of the ensembles has been connected to the quantum theory of optimal mass transport \cite{rouze2019concentration,carlen2020non,gao2020fisher,de2021quantum}.
In this context, a quantum generalization of the Lipschitz constant has been defined to quantify the maximum dependence of an observable on any given spin, and concentration inequalities have been proved for observables with $O(1)$ Lipschitz constant measured on weakly correlated states \cite{de2022quantum}.
The distinguishability between quantum states with respect to observables with unit Lipschitz constant has been quantified by the quantum Wasserstein distance of order $1$, or quantum $W_1$ distance \cite{de2021quantum,depalma2023quantum}.
Such concentration has been proved to be linked to transportation-cost inequalities that provide an upper bound to the $W_1$ distance between a generic state and a Gibbs state in terms of the square root of their relative entropy, and the equivalence between the canonical and microcanonical ensembles has been proved via such transportation-cost inequalities \cite[Proposition 12]{de2022quantum}.
It has further been proved that, whenever a Gibbs state satisfies such transportation-cost inequality, it is close with respect to the $W_1$ distance to any state that has approximately the same entropy, \emph{i.e.}, such that the difference between the entropy of the two states is $o(n)$, where $n$ is the number of spins.

In this paper, we consider the distinguishability with respect to $k$-local observables, \emph{i.e.}, the local observables such that each region $\Lambda$ in the sum \eqref{eq:local} contains at most $k=O(1)$ spins, and we quantify such distinguishability with the $k$-local quantum $W_1$ distance (\autoref{defn:kW1}).
First, we employ our concentration inequality for the states of a quantum spin lattice with exponentially decaying correlations to prove transportation-cost inequalities for the same states (\autoref{thm:bound_m_entropy}).
Then, we apply this result to prove that the Gibbs states with exponentially decaying correlations are close in the $k$-local quantum $W_1$ distance to any state with approximately the same entropy, \emph{i.e.}, such that the difference between the entropy of the two states is $o\left(n^\frac{1}{2d+1}\right)$ (\autoref{prop:equiv}), where $d>0$ depends on the growth of the size of the metric balls as explained above.
In particular, we prove the equivalence between the canonical and microcanonical ensembles for such Gibbs states with respect to all the $k$-local observables whose local norm is $O(1)$ (\autoref{thm:equiv}).

The paper is organized as follows.
In \autoref{sec:setup}, we set the notation for the paper.
In \autoref{sec:conc}, we introduce the local norm and prove the concentration inequalities for product states and states with exponentially decaying correlations.
In \autoref{sec:TCI}, we introduce the $k$-local quantum $W_1$ distance and prove that any quantum state with exponentially decaying correlations satisfies a transportation-cost inequality with respect to such distance.
In \autoref{sec:ensembles}, we introduce the problem of the equivalence of the ensembles of quantum statistical mechanics and prove such equivalence for the above states.
We conclude in \autoref{sec:concl}.
\autoref{sec:aux} contains the proof of the auxiliary lemmas.

\section{Setup}\label{sec:setup}

Let us start by setting the notation for the paper:

\sp

\begin{defn}
    For any $k\in\mathbb{N}$ we define
    \begin{equation}
        [k] = \{1,\,\ldots,\,k\}\,.
    \end{equation}
\end{defn}

We consider a quantum system made by $n$ spins, which we label with the integers from $1$ to $n$.
The local Hilbert space of each spin is $\mathbb{C}^q$, such that the Hilbert space of the system is $\left(\mathbb{C}^q\right)^{\otimes n}$.

\sp

\begin{defn}
    For any subset of the spins $\Lambda\subseteq[n]$, let
    \begin{equation}
        \mathcal{H}_\Lambda = \bigotimes_{x\in\Lambda}\mathbb{C}^q
    \end{equation}
    be the Hilbert space associated with the spins in $\Lambda$, let $\mathcal{O}_\Lambda$ be the set of the self-adjoint linear operators acting on $\mathcal{H}_\Lambda$, let $\mathcal{O}_\Lambda^T$ be the set of the traceless operators in $\mathcal{O}_\Lambda$, and let $\mathcal{S}_\Lambda$ be the set of the quantum states acting on $\mathcal{H}_\Lambda$.
\end{defn}

\sp

\begin{defn}[$k$-local operator]
    A linear operator $X$ acting on $\left(\mathbb{C}^q\right)^{\otimes n}$ is $k$-local if it is a linear combination of linear operators acting on at most $k$ spins each, \emph{i.e.}, if
    \begin{equation}
X = \sum_{\Lambda\subseteq[n]:|\Lambda|\le k}X_\Lambda\,,     
    \end{equation}
    where each $X_\Lambda$ acts only on the corresponding $\mathcal{H}_\Lambda$.
    For any $\Lambda\subseteq[n]$, we denote with $\mathcal{O}_\Lambda^{(k)}$ the set of the $k$-local self-adjoint linear operators acting on $\mathcal{H}_\Lambda$.
\end{defn}

\sp

When dealing with non-product states, we will need to introduce a distance $\mathsf{d}$ on the set of the spins $[n]$.
For example, if the spin are located at the vertices of a graph, $\mathsf{d}$ can be the graph distance given by the length of the minimum path connecting two vertices.
Otherwise, if the spins are located at points of $\mathbb{R}^d$, the distance can be the Euclidean distance or any $\ell^p$ distance.
We stress that we do not require any regular structure such as a square lattice, nor a notion of neighboring spins.

\sp

\begin{ass}[Dimensionality of the spin lattice]\label{ass:ball}
We assume that the size of the balls of $\mathsf{d}$ with radius $r$ are $O(r^d)$, \emph{i.e.}, that there exist $d\ge1$ and $A>0$ such that for any $v\in[n]$ and any $r>0$ we have
\begin{equation}\label{eq:bound_ball}
    |B_\mathsf d(v,r)|\leq \, A \,r^d\,,
\end{equation}
where $d$ plays the role of the dimensionality of the lattice.
\end{ass}

\sp

When dealing with non-product states, we will also need the notion of correlations decaying exponentially with the distance $\mathsf{d}$:

\sp

\begin{defn}[Correlation length]\label{ass:length}
We say that the state $\omega\in\mathcal{S}_{[n]}$ has correlation length $\xi>0$ if there exist $C>0$ independent of $n$ such that for any two disjoint regions $\Lambda_1,\,\Lambda_2\subset[n]$ and any two observables $h_{\Lambda_1}\in\mathcal{O}_{\Lambda_1}$ and $h_{\Lambda_2}\in\mathcal{O}_{\Lambda_2}$ we have
\begin{equation}
\left|\left\langle h_{\Lambda_1}\,h_{\Lambda_2}\right\rangle - \left\langle h_{\Lambda_1}\right\rangle\left\langle h_{\Lambda_2}\right\rangle\right| \le C\modifica{\left\|h_{\Lambda_1}\right\|\left\|h_{\Lambda_2}\right\|}\exp\left(-\frac{\mathsf d(\Lambda_1,\Lambda_2)}{\xi}\right)\,,
\end{equation}
where
\begin{equation}
\mathsf d(\Lambda_1,\Lambda_2) = \min\left\{\mathsf{d}(x_1,x_2) : x_1\in\Lambda_1,\,x_2\in\Lambda_2\right\}
\end{equation}
and the angle brackets denote the expectation with respect to $\omega$.
\end{defn}

\section{Concentration inequalities}\label{sec:conc}

In this section we introduce the local norm and derive the concentration inequalities for two relevant kinds of states of a spin lattice system: product states and states with exponentially decaying correlations.
More precisely, we bound from below the cumulative distribution functions of local observables with a function of their local norm.

\subsection{The local norm}
Let us introduce the local norm on which all our results are based:

\sp

\begin{defn}[Local norm]\label{local norm}\label{defn:local}
The local norm of an observable $H\in\mathcal{O}_{[n]}$ is
\begin{equation}\label{eq:loc}
\|H\|_{\textnormal{loc}} = 2\min\left\{\max_{x\in[n]}\sum_{\Lambda\ni x}\left\|h_\Lambda\right\| : H = \sum_{\Lambda\subseteq[n]}h_\Lambda\,,\; h_\Lambda\in\mathcal{O}_\Lambda\right\}\,,
\end{equation}
where $\|\cdot\|$ denotes the operator norm.
In words, we consider all the decompositions of $H$ as a sum of local operators, and define the dependence of any such decomposition on a spin $x$ as twice the sum of the operator norm of each local operator that acts on $x$.
We then define the local norm of such decomposition as the maximum dependence on a spin, and the local norm of $H$ as the minimum local norm of all its possible decompositions.
\end{defn}

\sp

\begin{rem}
    The local norm \eqref{eq:loc} is similar to the norm considered in \cite[Eq. (4)]{KUWAHARA2020gaussian}.
    The novelty of \eqref{eq:loc} with respect to \cite[Eq. (4)]{KUWAHARA2020gaussian} is that, while \cite[Eq. (4)]{KUWAHARA2020gaussian} considers only a fixed decomposition of $H$ of the form \eqref{eq:local}, \eqref{eq:loc} optimizes over all the possible decompositions. 
\end{rem}

\sp

\begin{rem}
    The local norm \eqref{eq:loc} can be recovered as a special case of the local norm defined in \cite{depalma2023classical}, where the norm of each $h_\Lambda$ is multiplied by a penalty $c_{|\Lambda|}$ depending on $|\Lambda|$:
    \begin{equation}\label{eq:locc}
\|H\|_{\widetilde{\textnormal{loc}}} = 2\min\left\{\max_{x\in[n]}\sum_{\Lambda\ni x}c_{|\Lambda|}\left\|h_\Lambda\right\| : H = \sum_{\Lambda\subseteq[n]}h_\Lambda\,,\; h_\Lambda\in\mathcal{O}_\Lambda\right\}\,.
\end{equation}
    The local norm \eqref{eq:loc} can be recovered from \eqref{eq:locc} by setting to one all the penalties.
\end{rem}

\sp

\begin{rem}
    $\|\,\,\|_{\textnormal{loc}}$ is a seminorm such that $\|I\|_{\textnormal{loc}}=0$. 
\end{rem}

\subsection{Product states}

We start by proving a bound on the cumulative distribution function of the local observable $H$ measured on a product state as a function of the local norm of $H$.

In order to derive a concentration inequality involving the quantum local norm, we will first prove an upper bound to the expectation value of $e^H$.
The concentration inequality will then follow from the Markov inequality. 

\sp

    \begin{thm}\label{thm:bound1}
     Let $\omega$ be a product state of $n$ spins.
     Let $H\in\mathcal{O}_{[n]}$ be an observable, and let
     \begin{equation}
         H = \sum_{\Lambda\subseteq[n]}\tilde{h}_\Lambda
     \end{equation}
     be the decomposition of $H$ that achieves the local norm as in \eqref{eq:loc}.
     For any $\emptyset\neq\Lambda\subseteq[n]$, let
     \begin{equation}\label{h_Lambda}
     h_\Lambda := \tilde{h}_\Lambda - \left\langle\tilde{h}_\Lambda\right\rangle\,,
     \end{equation}
     where the angle brackets denote the expectation with respect to $\omega$.
     Then,
     \begin{equation}\label{eq:concentration1}
        \left\langle e^H\right\rangle \leq  \exp\left[\left\langle H\right\rangle + \sum_{\Lambda\subseteq[n]} \left(e^{|\Lambda| \| H\|_{\textnormal{loc}}}-1-\frac{|\Lambda|\|H\|_{\textnormal{loc}}}{2}\right ) \| h_\Lambda\|\right].
     \end{equation}
    \end{thm}

  \begin{proof}
By definition of $h_\Lambda$, for every $\emptyset\not =\Lambda\subseteq[n]$ we have $\left\langle {h}_\Lambda\right\rangle = 0$ and
\begin{equation}\label{eq:HH_0}
H =\sum_{\Lambda\subseteq[n]}\tilde h_\Lambda= \tilde{h}_\emptyset + \sum_{\emptyset\neq\Lambda\subseteq[n]}\left(\left\langle\tilde{h}_\Lambda\right\rangle +h_\Lambda\right)= \left\langle H \right\rangle + \sum_{\emptyset\neq\Lambda\subseteq[n]}{h}_{\Lambda}\,.
  \end{equation}
        Let us define $H_0:=\sum_{\emptyset\neq\Lambda\subseteq[n]}{h}_{\Lambda}$ and observe that: 
        \begin{equation}
            \left\langle H_0^2\right\rangle =\sum_{\Lambda_1\cap\Lambda_2\not =\emptyset} \left\langle h_{\Lambda_1}h_{\Lambda_2}\right\rangle\leq\sum_{\Lambda_1\cap\Lambda_2\not =\emptyset} \| h_{\Lambda_1}\| \| h_{\Lambda_2}\|,
        \end{equation}since $\left\langle h_{\Lambda_1}h_{\Lambda_2}\right\rangle=\left\langle h_{\Lambda_1}\right\rangle\left\langle h_{\Lambda_2}\right\rangle=0$ for $\Lambda_1\cap\Lambda_2=\emptyset$. Therefore, in order to bound $\left\langle H_0^r\right\rangle$, we take the sum over the collections of $r$ regions $\Lambda_1,\dots,\Lambda_r$ such that any region does not have empty intersection with all the others:
\begin{equation}\label{eq:H^r}
    \left\langle H_0^r\right\rangle\leq \sum_{(\Lambda_1,...,\Lambda_r)\in \gR_r} \| h_{\Lambda_1}\|\cdots \|h_{\Lambda_r}\|,
\end{equation}
where $\gR_r=\{(\Lambda_1,...,\Lambda_r): \Lambda_i\cap (\cup_{j\not =i} \Lambda_j)\not =\emptyset, \forall i=1,...,r\}$.
For any graph $G=(V,E)$ with $V=\{1,...r\}$ let  
\begin{equation}\label{eq:C_G}
    C_G:=\sum_{\Lambda_1,...,\Lambda_r}\prod_{(i,j)\in E}|\Lambda_i\cap\Lambda_j| \| h_{\Lambda_1}\|\cdots \|h_{\Lambda_{r}}\|.
\end{equation}
Since $|\Lambda_i\cap\Lambda_j|\geq 1$ whenever $\Lambda_i\cap\Lambda_j\not =\emptyset$, we have $\left\langle H_0^r\right\rangle\leq \sum_{G\in\mathcal G_r} C_G$ where $\mathcal G_r=\{(V,E): |V|=r\, ,\, \deg(v)\geq 1\, \forall v\in V\}$. In words, we associate any non-empty intersection among two regions of the collections $\Lambda_1,...,\Lambda_r$ to an edge of a graph $G$ with $r$ vertices, in this sense we take the sum only over the graphs with no isolated vertices.  
Furthermore, the bound over $\left\langle H_0^r\right\rangle$  can be improved considering the {minimal graphs}. A {\em minimal graph} is defined as a graph such that the removal of any edge leaves at least one isolated vertex.
For any $(\Lambda_1,...,\Lambda_r)\in\gR_r$ there is a graph $G\in\mathcal G_r$ of $r$ vertices such that $(i,j)\in E$ if and only if $\Lambda_i\cap \Lambda_j\not =\emptyset$. Let $G'$ be a minimal subgraph of $G$, then:
\begin{equation}
    \|h_{\Lambda_1}\|\cdots\|h_{\Lambda_r}\|\leq \prod_{(i,j)\in E'}|\Lambda_i\cap\Lambda_j|\|h_{\Lambda_1}\|\cdots\|h_{\Lambda_r}\|.
\end{equation}

Now, let $\mathfrak G_r:=\{\textnormal{minimal graphs with $r$ vertices}\}$, for bounding $ \left\langle H_0^r\right\rangle$ we can take the sum only on the minimal graphs:

\begin{equation}\label{eq:H^r_1}
    \left\langle H_0^r\right\rangle \leq   \sum_{G\in\gG_r } C_G.
\end{equation}
Therefore:
\begin{equation}\label{eq:expH}
\left\langle e^{H_0} \right\rangle =\sum_{r=0}^{+\infty} \frac{\left\langle H_0^r\right\rangle}{r!}\leq \sum_{r=0}^{+\infty}  \sum_{G\in\gG_{r}} \frac{C_G}{r!} 
=\sum_{G\in\gG_{min}}  \frac{C_G}{V_G!},
\end{equation}
where $\gG_{min}$ is the set of all the minimal graph with a finite set of vertices and $V_G$ is the number of vertices of $G$.
We claim that any connected minimal graph is a star graph, \emph{i.e.}, has a vertex that belongs to all the edges.
Indeed, let $v$ be a vertex with maximum degree.
If such degree is $1$, then all vertices have degree $1$, and the graph is made by $2$ vertices connected by an edge.
Let the maximum degree be at least $2$.
Any vertex $w$ that is connected with $v$ by an edge cannot be connected to any other vertex, otherwise the edge $(v,w)$ could be removed without disconnecting any vertex.
Therefore, $v$ is connected by an edge to all the other vertices and the graph is a star graph.
Then, any minimal graph is the union of star graphs. Given a minimal graph with $V_G$ vertices, let us denote the number of connected components with $r$ vertices as $N_r$, then $V_G=\sum_{r\geq 2} rN_r$. The contribution $C_r$ in \eqref{eq:H^r_1} of the star graphs with $r$ vertices is:
\begin{equation}\label{eq:defC_r}
   C_r:= \sum_{\Lambda_1,\dots,\Lambda_r}|\Lambda_1\cap\Lambda_2|\cdots |\Lambda_1\cap \Lambda_r| \,\| h_{\Lambda_1}\|\cdots \|h_{\Lambda_r}\|,
\end{equation}
within the choice that $\Lambda_1$ corresponds to the core vertex. Let us consider a minimal graph $G$ formed by $N_r$ connected components, for $r\geq 2$. By definitions \eqref{eq:C_G} and \eqref{eq:defC_r}, we have:
\begin{equation}
    C_G=\prod_{r\geq 2} C_r^{N_r}.
\end{equation}
As a consequence, the last term in \eqref{eq:expH} can be re-written as a sum over the connected components:

\begin{align}\label{eq:sumC_G}
        \sum_{G\in\gG_{min}} \frac{C_G}{V_G!}
        =&\sum_{N_2,N_3,...}  \frac{1}{(2N_2+3N_3+\cdots)!} \times \frac{(2N_2+3N_3+\cdots)!}{N_2!2^{N_2}\,N_3!2!^{N_3}\,N_4!3!^{N_4}\cdots} \times{C_2^{N_2}C_3^{N_3}C_4^{N_4}\cdots}\nonumber\\
        =&\sum_{N_2}\frac{C_2^{N_2}}{2^{N_2}N_2!} \sum_{N_3,N_4...} \frac{C_3^{N_3}C_4^{N_4}\cdots}{\, 2!^{N_3}N_3! \,3!^{N_4}N_4!\,\cdots} = \exp\left[\frac{C_2}{2}+\sum_{r=3}^\infty \frac{C_r}{(r-1)!}\right].
\end{align} 
In \eqref{eq:sumC_G}, the factor $\frac{(2N_2+3N_3+\cdots)!}{N_2!2^{N_2}\,N_3!2!^{N_3}\,N_4!3!^{N_4}\cdots}$ is the number of different graphs corresponding to the fixed values of $N_i$ for every $i\geq 2$. In fact, a star graph with $r>2$ vertices is invariant w.r.t. all the $(r-1)!$ permutations of the external vertices (and a graph with 2 vertices is trivially invariant under the swap of its vertices). Moreover, a graph with $N_r$ connected components of $r$ vertices is invariant w.r.t. all the $N_r!$ permutations of these connected components. So, we take the quotient of the number of all the permutation of the $\sum_{r\geq 2} rN_r$ vertices over the total number of invariant operations.
Therefore, \eqref{eq:sumC_G} provides this bound: 
\begin{equation}\label{eq:C_r0}
    \left\langle e^{H_0} \right\rangle\leq \exp\left[\frac{C_2}{2}+\sum_{r=3}^\infty \frac{C_r}{(r-1)!}\right].
\end{equation}
By definition \eqref{eq:defC_r}, $C_r$ can be written in this form:
\begin{equation}\label{eq:3.11}
    C_r=\sum_{\Lambda_1,\dots,\Lambda_r}\sum_{v_2\in\Lambda_1\cap\Lambda_2}\cdots\sum_{v_r\in\Lambda_1\cap \Lambda_r} \| h_{\Lambda_1}\|\cdots \|h_{\Lambda_r}\|,
\end{equation}
because $\sum_{v_i\in\Lambda_1\cap\Lambda_i}\| h_{\Lambda_1}\|\cdots \|h_{\Lambda_r}\|=|\Lambda_1\cap\Lambda_i|\| h_{\Lambda_1}\|\cdots \|h_{\Lambda_r}\|$. Furthermore, since:
\begin{equation}
\sum_{\Lambda_i}\sum_{v_i\in\Lambda_1\cap\Lambda_i}\|h_{\Lambda_i}\|=\sum_{v_i\in\Lambda_1}\sum_{\Lambda_i\ni v_i}\|h_{\Lambda_i}\|\qquad \forall i=2,\dots,r,
\end{equation}
by \eqref{eq:3.11}, we have:
\begin{equation}\label{eq:C_r}
    C_r=\sum_{\Lambda_1} \sum_{v_2,...,v_r\in\Lambda_1}\sum_{\Lambda_2\ni v_2}\dots \sum_{\Lambda_r\ni v_r} \| h_{\Lambda_1}\|\cdots \|h_{\Lambda_r}\|
    = \sum_{\Lambda_1}\|h_{\Lambda_1}\|\left(\sum_{v_2\in\Lambda_1}\sum_{\Lambda_2\ni v_2} \|h_{\Lambda_2}\|\right)^{r-1}.
\end{equation}

In view of \autoref{lem:loc} we have $\sum_{\Lambda_2\ni v_2} \|h_{\Lambda_2}\|\leq \|H\|_{\textnormal{loc}}$ and observing that
\begin{equation}
    \sum_{v_2\in\Lambda_1}\|H\|_{\textnormal{loc}}=|\Lambda_1|\|H\|_{\textnormal{loc}}\,,
\end{equation}
from \eqref{eq:C_r} we obtain:
\begin{equation}\label{eq:C_r_bound}
    C_r\leq\sum_{\Lambda_1}\|h_{\Lambda_1}\|\left(\sum_{v_2\in\Lambda_1}\sum_{\Lambda_2\ni v_2} \|h_{\Lambda_2}\|\right)^{r-1}\leq\sum_{\Lambda\subseteq[n]} |\Lambda|^{r-1} \| H\|^{r-1}_{\textnormal{loc}} \| h_\Lambda\|\qquad \forall r\geq 2,
\end{equation}
then, by substituting in \eqref{eq:C_r0}:
\begin{align}
\left\langle e^{H_0} \right\rangle & \leq \exp\left[\frac{C_2}{2} +\sum_{r=3}^\infty \sum_{\Lambda\subseteq[n]} \frac{|\Lambda|^{r-1} \| H\|^{r-1}_{\textnormal{loc}} \| h_\Lambda\|}{(r-1)!} \right]= \\ \nonumber
&=\exp\left[\frac{C_2}{2} +\sum_{\Lambda\subseteq[n]} \left(e^{|\Lambda| \| H\|_{\textnormal{loc}}}-1-|\Lambda|\|H\|_{\textnormal{loc}}\right ) \| h_\Lambda\|\right]  \leq\\ \nonumber
& \leq \exp\left[\sum_{\Lambda\subseteq[n]} \left(e^{|\Lambda| \| H\|_{\textnormal{loc}}}-1-\frac{|\Lambda|\|H\|_{\textnormal{loc}}}{2}\right ) \| h_\Lambda\|\right],
\end{align}
where, in the last inequality, we have used the fact $C_2\leq\sum_{\Lambda\subseteq[n]}|\Lambda|\|H\|_{\textnormal{loc}} \| h_\Lambda\|$. The claim follows from $\left\langle e^H\right\rangle=e^{\left\langle H\right\rangle}\left\langle e^{H_0}\right\rangle$.
  \end{proof}  

The inequality \eqref{eq:concentration1} stated by the theorem above provides a bound that is not a function of the local norm of $H$ alone. However, such a bound can be obtained with the additional requirement that the size of the regions, on which $H$ acts locally, satisfies $|\Lambda|\leq k$.

\sp

\begin{cor}\label{prop:expH}
    Under the hypotheses of \autoref{thm:bound1}, if $H$ is $k$-local then
    \begin{equation}\label{eq:expH2}
        \left\langle e^{H}\right\rangle \leq \exp\left[\left\langle H\right\rangle+\frac{n}{k}\left({e^{k \| H\|_{\textnormal{loc}}}-1-\frac{k\|H\|_{\textnormal{loc}}}{2}}\right)   \| H\|_{\textnormal{loc}}\right]
    \end{equation}
\end{cor}

\begin{proof}
    By \autoref{thm:bound1}, we have:
      \begin{align}\label{eq:thm3.2}
                  \left\langle e^H\right\rangle & \leq \exp\left[\left\langle H\right\rangle+\sum_{\Lambda\subseteq[n]}\sum_{v\in\Lambda} \frac{1}{|\Lambda|}\left(e^{|\Lambda| \| H\|_{\textnormal{loc}}}-1-\frac{|\Lambda|\|H\|_{\textnormal{loc}}}{2}\right ) \| h_\Lambda\|\right] \\ \nonumber
        &\leq \exp\left[\left\langle H\right\rangle+ \frac{1}{k}\left(e^{k \| H\|_{\textnormal{loc}}}-1-\frac{k\|H\|_{\textnormal{loc}}}{2}\right ) \sum_{\Lambda\subseteq[n]}\sum_{v\in\Lambda} \| h_\Lambda\|\right]\\ \nonumber
        &= \exp\left[\left\langle H\right\rangle+ \frac{1}{k}\left(e^{k \| H\|_{\textnormal{loc}}}-1-\frac{k\|H\|_{\textnormal{loc}}}{2}\right ) \sum_{v\in[n]}\sum_{\Lambda \ni v} \| h_\Lambda\|\right]\\ \nonumber
        &\leq \exp\left[\left\langle H\right\rangle+ \frac{1}{k}\left(e^{k \| H\|_{\textnormal{loc}}}-1-\frac{k\|H\|_{\textnormal{loc}}}{2}\right ) \sum_{v\in[n]} \| H\|_{\textnormal{loc}}\right]\\ \nonumber
        & = \exp\left[\left\langle H\right\rangle+ \frac{1}{k}\left(e^{k \| H\|_{\textnormal{loc}}}-1-\frac{k\|H\|_{\textnormal{loc}}}{2}\right ) n \| H\|_{\textnormal{loc}}\right]\,.
      \end{align}
The first inequality in \eqref{eq:thm3.2} is implied by the hypothesis $| \Lambda|\leq k$ and the second one is a consequence of $\sum_{\Lambda\ni v}\|h_\Lambda\|\leq \|H\|_{\textnormal{loc}}$ for every site $v$.
\end{proof}

\autoref{prop:expH} can be used to derive an exponential concentration inequality via the inequality
\begin{equation}\label{eq:Markov}
    \mathbb P(H\geq na)\leq e^{-t na} \left\langle e^{tH}\right\rangle\qquad \forall\;t,\,a>0\,:
\end{equation}

\sp

\begin{thm}\label{thm:concentrationproduct}
    Let $\omega$ be a product state of $n$ spins and $H$ be a $k$-local Hamiltonian. Then, for any $a>0$ we have
    \begin{equation}\label{eq:thm3.3}
         \mathbb P\left(H\geq \left\langle H\right\rangle + na\right)\leq \exp\left[-\frac{n}{k^2} \,F\left(\frac{ak}{\|H\|_{\textnormal{loc}}}\right)   \right],
    \end{equation}
    where $F(x):=\max_{s>0} s(x- e^s+1+\frac{s}{2})$.
\end{thm}

\begin{proof}
Without loss of generality, we can assume $\left\langle H\right\rangle=0$.
    First, we apply \autoref{prop:expH} inserting \eqref{eq:expH2} into \eqref{eq:Markov}:
    \begin{equation}
        \mathbb P(H\geq na)\leq\exp\left[-tna +\frac{nt}{k}\left(e^{tk\|H\|_{\textnormal{loc}}}-1-\frac{tk\|H\|_{\textnormal{loc}}}{2}\right)\|H\|_{\textnormal{loc}}   \right].
    \end{equation}
Let us apply the substitution $s=tk\|H\|_{\textnormal{loc}}$: 
 \begin{equation}\label{eq:P}
        \mathbb P(H\geq na)\leq\exp\left[-\frac{n}{k^2} s\left(\frac{ak}{\|H\|_{\textnormal{loc}}} - e^s+1+\frac{s}{2}  \right)  \right]
    \end{equation}
    and, in turn, perform the minimization over $s$ in the right-hand side of \eqref{eq:P} obtaining \eqref{eq:thm3.3}.
\end{proof}

\begin{rem}
    The function $F$ in \eqref{eq:thm3.3} is strictly increasing with $F(0)=0$, then it is strictly positive for any $x>0$. 
\end{rem}

\vspace{0.25cm}

\modifica{

\begin{rem}\label{rem_1}
    The function $F$ in \eqref{eq:thm3.3} is defined to keep the bound as tight as possible. By definition of $F$, the maximizer $s^*$ of $f(s)=s(x-e^s+1+s/2)$ is a function of $x\geq 0$ which does not admit an explicit form in terms of elementary functions. However, one can suppose to approximate $s^*$ by a suitable elementary function in order to explicitly obtain a function $\widetilde F:[0,+\infty)\rightarrow [0,+\infty)$, with $\widetilde F(x)\le F(x)$ for every $x\geq 0$, which provides a suboptimal but more explicit bound in \autoref{eq:thm3.3}. 
    %For instance, since $F$ is superpolynomial, we can determine a polynomial function which definitely minorates $F$ for any $x>B$ for some $B>0$. 
    For any fixed $x\geq 0$, the optimum $s^*$ of $f$ satisfies:
    $$x=(e^{s^*}-1)(s^*+1).$$
    A rough approximation of the inverse of $x(s^*)=(e^{s^*}-1)(s^*+1)$ is $\hat s^*(x)=\log(x+1)$ that provides the {\em suboptimal} version of $F$ by:
    $$\widetilde F(x)=f(\hat s^*(x))=\log(x+1)\left[x-e^{\log(x+1)}+1+\frac{1}{2}\log(x+1)\right]=\frac{1}{2}\log^2(x+1)$$
    Then, we get a worse but more explicit bound than the one given in \autoref{eq:thm3.3}:
    \begin{equation}
         \mathbb P\left(H\geq \left\langle H\right\rangle + na\right)\leq \exp\left[-\frac{n}{2k^2} \,\log^2\left(\frac{ak}{\|H\|_{\textnormal{loc}}}+1\right)   \right],
    \end{equation}
    so we can easily note that the concentration weakens as the locality $k$ of $H$ increases. 
\end{rem}}

\modifica{
\subsubsection{Comparison with \cite[Corollary 3]{de2021quantum}}
Ref. \cite{de2021quantum} proves the following concentration inequality for the maximally mixed state:

\sp

\begin{thm}[{\cite[Corollary 3]{de2021quantum}}]
    Let $\omega$ be the maximally mixed state of $n$ spins. Then, for any $H\in\mathcal{O}_{[n]}$ and any $a>0$ we have
    \begin{equation}\label{eq:concL}
         \mathbb P\left(H\geq \left\langle H\right\rangle + na\right)\leq \exp\left[-\frac{2\,n\,a^2}{\left\|H\right\|_L^2}\right],
    \end{equation}
    where
    \begin{equation}
        \left\|H\right\|_L = 2\max_{x\in[n]}\min\left\{\left\|H - H_{[n]\setminus\{x\}}\otimes\mathbb{I}_x\right\| : H_{[n]\setminus\{x\}}\in\mathcal{O}_{H_{[n]\setminus\{x\}}}\right\}
    \end{equation}
    is the quantum Lipschitz constant of $H$.
\end{thm}

\sp

For the maximally mixed state, \cite[Corollary 3]{de2021quantum} is stronger than \autoref{thm:concentrationproduct}.
Indeed, both results imply a concentration that is exponential in $n$ whenever $\left\|H\right\|_L$ and $\left\|H\right\|_{\mathrm{loc}}$, respectively, are $O(1)$.
From \cite[Proposition II.5]{depalma2023classical} we have $\left\|H\right\|_L \le \left\|H\right\|_{\mathrm{loc}}$, hence \cite[Corollary 3]{de2021quantum} guarantees an exponential concentration for a larger family of Hamiltonians.
However, the fundamental improvement of \autoref{thm:concentrationproduct} is that it is valid not only for the maximally mixed state but for any product state, including pure product states.
The validity of \eqref{eq:concL} for generic product states remains an open challenge.
}

\subsection{Quantum states with finite correlation length}\label{ssec:finitecl}
We now consider the quantum states with exponentially decaying correlations characterized by a finite correlation length $\xi$ in the sense of \autoref{ass:length}. Also in this case, we consider the decomposition of a local Hamiltonian $H$ that achieves the local norm $H=\sum_{\Lambda\subseteq[n]} \tilde h_\Lambda$ and define $h_\Lambda := \tilde{h}_\Lambda - \left\langle\tilde{h}_\Lambda\right\rangle$ for any $\emptyset\neq\Lambda\subseteq[n]$. 
Moreover, under \autoref{ass:ball}, for any region $\Lambda$, let us define its \emph{enlargement} by $l>0$ as $\widetilde\Lambda :=\{x: \mathsf d(x,\Lambda)\leq l\}$.
We can now apply to the enlarged regions the argument that we used for product states. In the new case, for $\widetilde\Lambda_1\cap \widetilde\Lambda_2 =\emptyset$ we have
\begin{equation}
|\left\langle h_{\Lambda_1}h_{\Lambda_2}\right\rangle|\leq \left\|h_{\Lambda_1}\right\|\left\|h_{\Lambda_2}\right\|C\,e^{-\frac{\mathsf d(\Lambda_1,\Lambda_2)}{\xi}}\leq \left\|h_{\Lambda_1}\right\|\left\|h_{\Lambda_2}\right\|C\,e^{-\frac{l}{\xi}}\,,
\end{equation}
where $C>0$ and $\xi$ is the correlation length. Therefore, for a collections of regions $\Lambda_1,...,\Lambda_r$ such that $\widetilde\Lambda_i\cap (\cup_{j\not =i} \widetilde\Lambda_j) =\emptyset$ for some $i\in[r]$, we have
\begin{equation}\label{eq:bound_xi}
    |\left\langle h_{\Lambda_1}\cdots h_{\Lambda_r}\right\rangle|\leq \| h_{\Lambda_1}\|\cdots \|h_{\Lambda_r}\| C e^{-\frac{l}{\xi}}.
\end{equation}

\begin{thm}\label{thm:bound2}
     Let $\omega\in\mathcal S_{[n]}$ be a state of a quantum system of $n$ spins with correlation length $\xi$ as in \autoref{ass:length} with respect to a distance satisfying \autoref{ass:ball}.
     Then, for any $k$-local Hamiltonian $H$ we have
     \begin{equation}\label{eq:e^Hcorr}
        \left\langle e^H\right\rangle \leq \exp\left[\left\langle H\right\rangle+\frac{n}{k}\left({e^{k\, (Al^{d})^2 \| H\|_{\textnormal{loc}}}-1-\frac{k \,(Al^{d})^2\|H\|_{\textnormal{loc}}}{2}}\right)   \| H\|_{\textnormal{loc}}\right]+Ce^{\left\langle H\right\rangle+n\|H\|_{\textnormal{loc}}-\frac{l}{\xi}}\,,
 \end{equation}
 where the angle brackets denote the expectation with respect to $\omega$.
    \end{thm}

\begin{proof}
   Let $H=\sum_{\Lambda\subseteq[n]} \tilde h_\Lambda$ be the decomposition of $H$ that achieves the local norm as in \eqref{eq:loc} and define $h_\Lambda := \tilde{h}_\Lambda - \left\langle\tilde{h}_\Lambda\right\rangle$ for any $\emptyset\neq\Lambda\subseteq[n]$ and $H_0:=\sum_{\emptyset\neq\Lambda\subseteq[n]} h_\Lambda$. Then, \eqref{eq:HH_0} is still valid and $\left\langle e^H\right\rangle=e^{\left\langle H\right\rangle}\left\langle e^{H_0}\right\rangle$. Adapting the inequality \eqref{eq:H^r} to enlarged regions $\{\widetilde \Lambda\}_{\Lambda\subseteq [n]}$ and considering \eqref{eq:bound_xi}, the expectation of $H_0^r$ can be bounded as follows:
  \begin{equation}\label{eq:H^r2}
      \left\langle H_0^r\right\rangle \leq \sum_{(\Lambda_1,\dots,\Lambda_r)\in\widetilde\gR_r} \|h_{\Lambda_1}\|\cdots\|h_{\Lambda_r}\|+\sum_{(\Lambda_1,\dots,\Lambda_r)\not\in\widetilde\gR_r}\|h_{\Lambda_1}\|\cdots\|h_{\Lambda_r}\| Ce^{-\frac{l}{\xi}},
  \end{equation}
where $\widetilde\gR_r=\{(\Lambda_1,...,\Lambda_r): \widetilde\Lambda_i\cap (\cup_{j\not =i} \widetilde\Lambda_j)\not =\emptyset\;\forall\, i=1,...,r  \}$. Therefore: 
\begin{equation}\label{eq:e^H2}
    \left\langle e^{H_0}\right\rangle\leq \sum_{r=0}^{+\infty}\frac{1}{r!}\sum_{(\Lambda_1,\dots,\Lambda_r)\in\widetilde\gR_r} \|h_{\Lambda_1}\|\cdots\|h_{\Lambda_r}\|+\sum_{r=0}^{+\infty}\frac{1}{r!}\sum_{(\Lambda_1,\dots,\Lambda_r)\not\in\widetilde\gR_r}\|h_{\Lambda_1}\|\cdots\|h_{\Lambda_r}\| Ce^{-\frac{l}{\xi}}.
\end{equation}
Now, let us focus on the sum over $(\Lambda_1,\dots,\Lambda_r)\in\widetilde\gR_r$ in \eqref{eq:e^H2} and define:
\begin{equation}\label{}
   C_r:= \sum_{\Lambda_1,\dots,\Lambda_r}|\widetilde\Lambda_1\cap\widetilde\Lambda_2|\cdots |\widetilde\Lambda_1\cap \widetilde\Lambda_r| \,\| h_{\Lambda_1}\|\cdots \|h_{\Lambda_r}\|.
\end{equation}
so that we can adapt the inequality \eqref{eq:C_r_bound}:  
\begin{equation}
   C_r\leq\sum_{\Lambda_1}\|h_{\Lambda_1}\|\left(\sum_{v_2\in\widetilde\Lambda_1}\sum_{\Lambda_2: v_2\in\widetilde\Lambda_2} \|h_{\Lambda_2}\|\right)^{r-1}.
\end{equation}
In general, we have that $v\in\widetilde\Lambda$ if and only if $\Lambda \cap B_\mathsf d (v,l) \not = \emptyset$ where $B_\mathsf d (v,l)$ is the ball centered in $v$ with radius $l$ referred to the distance $\mathsf d$ used to define the enlargement of the regions. Therefore:
\begin{equation}\label{eq:B(v,l)}
    \sum_{\Lambda:v\in\widetilde\Lambda} \|h_{\Lambda}\|=\sum_{\Lambda:\Lambda \cap B_\mathsf d (v,l) \not = \emptyset} \|h_{\Lambda}\|\leq \sum_{\Lambda\subseteq[n]} |\Lambda\cap B_\mathsf d(v,l)|\,\|h_\Lambda\|,
\end{equation}
the last term in \eqref{eq:B(v,l)} can be rewritten as:
\begin{equation}
    \sum_{\Lambda\subseteq[n]} \sum_{w\in\Lambda\cap B_\mathsf d(v,l)} \|h_\Lambda\| =\sum_{w\in B_\mathsf d(v,l)}\sum_{\Lambda\ni w} \|h_\Lambda\|\leq |B_\mathsf d(v,l)|\, \|H\|_{\textnormal{loc}}\leq A\, l^d \,\|H\|_{\textnormal{loc}},
\end{equation}
where, in the last inequality, we have taken into account \autoref{ass:ball}. This results in the following variant of \eqref{eq:C_r_bound}:
\begin{equation}
    C_r\leq \sum_{\Lambda\subseteq[n]} \|h_\Lambda\|\,[(A l^{d})^2 |\Lambda| \|H\|_{\textnormal{loc}}]^{r-1}\qquad \forall r=2,\dots,r,
\end{equation}
where we have also used the fact that $|\widetilde\Lambda|\leq |\Lambda| A\, l^d$ for any $\Lambda\subseteq [n]$. The latter is a consequence of \eqref{eq:bound_ball} since $\widetilde\Lambda\subset \cup_{v\in\Lambda} B_{\mathsf d}(v,l)$, so $|\widetilde\Lambda|\leq \sum_{v\in\Lambda}|B_{\mathsf d}(v,l)|\leq \sum_{v\in\Lambda} A l^d=|\Lambda| A l^d$. Therefore we can state:
\begin{align}    \label{eq:in_gR}
\sum_{r=0}^\infty \frac{1}{r!} \!\!\!\! \sum_{\quad(\Lambda_1,\dots,\Lambda_r)\in\widetilde\gR_r} &\|h_{\Lambda_1}\|\cdots\|h_{\Lambda_r}\|\leq \\\nonumber
    &\leq\exp\left[\sum_{\Lambda\subseteq[n]} \left(e^{(A l^d)^2|\Lambda| \| H\|_{\textnormal{loc}}}-1-\frac{(A l^d)^2|\Lambda|\|H\|_{\textnormal{loc}}}{2}\right ) \| h_\Lambda\|\right],
\end{align}
requiring that $|\Lambda|\leq k$ for any $\Lambda\in[n]$, the proof of \autoref{prop:expH} can be repeated yielding: 
\begin{align}    
\sum_{r=0}^\infty \frac{1}{r!} \!\!\!\! \sum_{\quad(\Lambda_1,\dots,\Lambda_r)\in\widetilde\gR_r} &\|h_{\Lambda_1}\|\cdots\|h_{\Lambda_r}\|\leq \\\nonumber
    &\leq \exp\left[\frac{n}{k}\left({e^{k (Al^d)^2 \| H\|_{\textnormal{loc}}}-1-\frac{k (Al^d)^2\|H\|_{\textnormal{loc}}}{2}}\right)   \| H\|_{\textnormal{loc}}\right].
\end{align}
Now, let us bound the sum over $(\Lambda_1,\dots,\Lambda_r)\not\in\widetilde\gR_r$ in \eqref{eq:H^r2}:
\begin{equation}
    \sum_{(\Lambda_1,\dots,\Lambda_r)\not\in\widetilde\gR_r}\|h_{\Lambda_1}\|\cdots\|h_{\Lambda_r}\| C e^{-\frac{l}{\xi}}\leq  \left(\sum_{\Lambda\subseteq[n]} \|h_\Lambda\|\right)^r C e^{-\frac{l}{\xi}}\leq \left(n\| H\|_{\textnormal{loc}}\right)^r C e^{-\frac{l}{\xi}},
\end{equation}
then: 
\begin{equation} \label{eq:not_in_gR}
    \sum_{r=0}^\infty \frac{1}{r!}\sum_{(\Lambda_1,\dots,\Lambda_r)\not\in\widetilde\gR_r}\|h_{\Lambda_1}\|\cdots\|h_{\Lambda_r}\| C e^{-\frac{l}{\xi}}\leq \sum_{r=0}^{+\infty}\frac{\left(n\| H\|_{\textnormal{loc}}\right)^r}{r!} C e^{-\frac{l}{\xi}}\leq e^{n\|H\|_{\textnormal{loc}}}C e^{-\frac{l}{\xi}}.
\end{equation}
Taking into account \eqref{eq:e^H2}, \eqref{eq:in_gR}, and \eqref{eq:not_in_gR} we have:
 \begin{equation}
        \left\langle e^{H_0}\right\rangle \leq \exp\left[\frac{n}{k}\left({e^{k\, (Al^{d})^2 \| H\|_{\textnormal{loc}}}-1-\frac{k \,(Al^{d})^2\|H\|_{\textnormal{loc}}}{2}}\right)   \| H\|_{\textnormal{loc}}\right]+Ce^{n\|H\|_{\textnormal{loc}}-\frac{l}{\xi}},
 \end{equation}
that directly implies the claim.
\end{proof}

As in the case of product states, \autoref{thm:bound2} implies the following concentration inequality for the quantum states with finite correlation length:

\sp

\begin{cor}\label{thm:concentration2}
     Let $\omega$ be a state of $n$ quantum spins with correlation length $\xi$ as in \autoref{ass:length} with respect to a distance satisfying \autoref{ass:ball}.
     Then, for any $k$-local Hamiltonian $H$ we have
     \begin{equation}\label{eq:thm:concentration2}
      \mathbb P(H\geq \left\langle H\right\rangle + na)\leq \left(C+1\right)\exp \left[-n^\frac{1}{2d+1} \,\frac{F\left(\frac{ak}{\|H\|_{\textnormal{loc}}}\right)}{k^2A^2\left[\frac{\xi}{kA^2}\left(k\,s^*\left(\frac{ak}{\|H\|_{\textnormal{loc}}}\right)+ F\left(\frac{ak}{\|H\|_{\textnormal{loc}}}\right)\right)  \right]^{\frac{2d}{2d+1}} }\right]\,,
 \end{equation}
 where $F(x):=\max_ss(x-e^s+1+\frac{s}{2})$ and $s^*(x) := \mathrm{argmax}_s s(x-e^s+1+\frac{s}{2})$.
\end{cor}

\begin{proof}
Without loss of generality, we can assume $\left\langle H\right\rangle=0$.
    Inserting the bound \eqref{eq:e^Hcorr} into the Markov inequality \eqref{eq:Markov}, we have
    \begin{align}
    \mathbb P(H\geq na) &\leq \exp\left[-tna +\frac{nt}{k}\left(e^{tkA^2l^{2d}\|H\|_{\textnormal{loc}}}-1-\frac{tkA^2l^{2d}\|H\|_{\textnormal{loc}}}{2} \right) \|H\|_{\textnormal{loc}} \right]\nonumber\\
    &\phantom{\le} + C\exp\left[nt\|H\|_{\textnormal{loc}}-\frac{l}{\xi}  \right].
    \end{align}
Following the approach adopted in the proof of \autoref{thm:concentrationproduct}, once introduced the variable $s=tkA^2l^{2d}\|H\|_{\textnormal{loc}}$, the concentration inequality becomes
 \begin{equation}\label{eq:conctr(s,l)}
         \mathbb P(H\geq na) \leq \exp\left[-\frac{n}{k^2A^2 l^{2d}} s\left(\frac{ak}{\|H\|_{\textnormal{loc}}}-e^s+1+\frac{s}{2}\right)  \right]+ C\exp\left[\frac{ns}{k A^2\, l^{2d}}-\frac{l}{\xi}  \right].
     \end{equation}
In addition to the minimization over $s$ the right-hand side of the above inequality can be optimized also w.r.t. $l$. Instead of performing such a double optimization, let us improve the bound on $\mathbb P(H\geq na)$ minimizing over $s$ only the first exponential term in \eqref{eq:conctr(s,l)} with the same strategy of the proof of \autoref{thm:concentrationproduct}.
We have
\begin{equation}\label{eq:conctr(s*,l)}
   \mathbb P(H\geq na) \leq \exp\left[-\frac{n}{k^2A^2 l^{2d}} F\left(\frac{ak}{\|H\|_{\textnormal{loc}}}\right)  \right]+ C\exp\left[\frac{n\,s^*\left(\frac{ak}{\|H\|_{\textnormal{loc}}}\right)}{k A^2\, l^{2d}}-\frac{l}{\xi}  \right]. 
\end{equation}
Now, let us fix the value of $l$ requiring the balance of the two exponential terms:
\begin{equation}
-n F\left(\frac{ak}{\|H\|_{\textnormal{loc}}}\right)=kn\,s^*\left(\frac{ak}{\|H\|_{\textnormal{loc}}}\right)-\frac{l^{2d+1}k A^2}{\xi},
\end{equation}
then:
\begin{equation}\label{eq:l}
    l=\left[\frac{\xi}{kA^2}\left(kn\,s^*\left(\frac{ak}{\|H\|_{\textnormal{loc}}}\right)+n F\left(\frac{ak}{\|H\|_{\textnormal{loc}}}\right)\right)  \right]^{\frac{1}{2d+1}}.
\end{equation}
By the substitution into \eqref{eq:conctr(s*,l)} we obtain:
\begin{equation}
    \mathbb P(H\geq na)\leq (C+1)\exp\left[-\frac{n^{\frac{1}{2d+1}}F\left(\frac{ak}{\|H\|_{\textnormal{loc}}}\right)}{k^2A^2\left[\frac{\xi}{kA^2}\left(k\,s^*\left(\frac{ak}{\|H\|_{\textnormal{loc}}}\right)+ F\left(\frac{ak}{\|H\|_{\textnormal{loc}}}\right)\right)  \right]^{\frac{2d}{2d+1}} } \right],
\end{equation}
    which is the claim.
\end{proof}

\begin{rem}
In the proof of \autoref{thm:concentration2}, we have improved the bound of inequality \eqref{eq:conctr(s,l)} optimizing over $s$ and fixing the value of $l$ by the condition that the two exponential terms are balanced. However, even if these two steps are swapped, the obtained bound is the same as one can check by direct inspection.
\end{rem}

\vspace{0.25cm}

\modifica{\begin{rem}
We can write a concentration inequality that is expressed in terms of elementary functions and is more explicit than \eqref{eq:thm:concentration2}, at the cost of worsening the bound. In the proof of \autoref{thm:concentration2}, the inequality \eqref{eq:conctr(s*,l)} is obtained by optimizing over $s$.
However, we can choose the approximate optimum $\hat s^*(x)=\log(x+1)$ as done in \autoref{rem_1}. Leaving the selection of the value of $l$ given by \eqref{eq:l} unchanged, the resulting bound is: 
\begin{equation}
    \mathbb P(H\geq na)\leq (C+1)\exp\left[-\frac{n^{\frac{1}{2d+1}}\log^2\left(\frac{ak}{\|H\|_{\textnormal{loc}}}+1\right)}{2k^2A^2\left[\frac{\xi}{kA^2}\left(k\,\log\left(\frac{ak}{\|H\|_{\textnormal{loc}}}+1\right)+ \frac{1}{2}\log^2\left(\frac{ak}{\|H\|_{\textnormal{loc}}}+1\right)\right)  \right]^{\frac{2d}{2d+1}} } \right].
\end{equation}

\end{rem}}

\modifica{
\subsubsection{Comparison with previous concentration bounds}

Our concentration bound for states with exponential decay of correlations, \autoref{thm:concentration2}, can be compared with existing results. In particular, under the same assumption on the correlation decay, in \cite{kuwahara2020eigenstate}, the authors provide the following inequality: 
\begin{equation}
    \mathbb P(H\geq \langle H\rangle + na )\geq \min\left(1, (e+3e\xi)\max\left(e^{-(\frac{na^2}{c})^{\frac{1}{d+1}}}, e^{-\frac{na^2}{c'k^d }}\right)\right),
\end{equation}
where $c$ and $c'$ are constant depending on $d$ and $\xi$, and $k$ is the locality of $H$. A similar order of decay is provided by \cite[Theorem 4.2]{anshu2016concentration}. However, these works have the requirement that the diameter of the regions, where the local terms of the observables act, is $O(1)$ and the spin lattice presents a regular structure. 

In \cite{de2022quantum}, there are some results about Gaussian concentration of the values of local observables. In particular \cite[Theorem 7]{de2022quantum} provides a concentration bound for states satisfying a transportation cost inequality, defined by the quantum $W_1$ distance as defined in \cite{de2021quantum}. In the particular case where the state commutes with the considered observable $H$, the concentration assumes the form: 
\begin{equation}
    \mathbb P(H\geq \langle H\rangle +na)\leq 2\exp\left(-\frac{(na)^2}{c\|H\|^2_L}\right)\,,
\end{equation}
where the constant $c$ is linked to a transportation cost inequality satisfied by the Gibbs state and is proved to be $O(n)$ at high enough temperature.
However, \cite[Theorem 7]{de2022quantum} requires that $H$ is a sum of commuting local terms.

In \autoref{table}, we compare both the concentration bounds provided by \autoref{thm:concentrationproduct} and \autoref{thm:concentration2} with some known results. In particular we highlight the assumptions under which the bounds are derived in order to stress the novelty of our concentration inequalities.

\begin{table}[h]
\centering
\begin{tabular}{ |c| c | c | c | }
\hline
\textbf{Concentration} & \textbf{Quantum states}& \textbf{Assumptions} &\textbf{Ref.} \\ 
 \hline
 $e^{-\Omega(n)}$ & Product states & $\|H\|_{\mathrm{loc}}$ is $O(1)$ & Our result \\
 & & Any $\Lambda$ contains $O(1)$ spins  &(\autoref{thm:concentrationproduct})\\
 \hline
 $e^{-\Omega(n)}$ & Maximally mixed state & $\|H\|_L$ is $O(1)$ & \cite[Corollary 3]{depalma2023quantum} \\
 \hline
 $e^{-\Omega(n)}$ & Product states & Regions overlap is $O(1)$  & \cite[Theorem 5.2]{anshu2016concentration} \\
 \hline
 %Nearest Neighbors& Grover & &Quadratic\\
 %\hline
 $\exp\left(-\Omega\left(n^{\frac{1}{2d+1}}\right)\right)$ & States with finite  & $\|H\|_{\mathrm{loc}}$ is $O(1)$ & Our result \\ 
& correlation length & Any $\Lambda$ contains $O(1)$ spins  & (\autoref{thm:concentration2})\\

 \hline
 $\exp\left(-\Omega\left(n^{\frac{1}{2d+1}}\right)\right)$  & States with finite & Diameter of $\Lambda$ is $O(1)$ & \cite[Equation S.17]{kuwahara2020eigenstate} \\
 & correlation length & Regular spin lattice & \cite[Theorem 4.2]{anshu2016concentration}\\
 \hline
 $e^{-\Omega(n)}$ & High-temperature states & Commuting local terms in $H$ & \cite[Theorem 7]{de2022quantum}\\
 \hline
\end{tabular}

\sp

\caption{In the table, we report the concentration bounds presented in this paper and some previous known results. The first column reports the Bound over $\mathbb P(H\geq \langle H\rangle+na)$, the second one specifies the considered quantum states, and the third column contains the further assumptions under which the bounds were derived. The regions $\Lambda$ are those involved in the definition of the local observable $H$.}
\label{table}
\end{table}
}

\section{A transportation-cost inequality}\label{sec:TCI}

\subsection{The \texorpdfstring{$k$}{k}-local quantum \texorpdfstring{$W_1$}{W1} distance}\label{sec:distance}
We define the $k$-local quantum $W_1$ distance as the distance that quantifies the distinguishability of quantum states with respect to the $k$-local observables with unit local norm:

\sp

\begin{defn}[$k$-local quantum $W_1$ norm]\label{defn:kW1}
    We define the $k$-local quantum $W_1$ norm as the norm on $\mathcal{O}_{[n]}^T$ that is dual to the local norm on $\mathcal{O}_{[n]}^{(k)}$: For any $\Delta\in\mathcal{O}_{[n]}^T$,
    \begin{equation}\label{eq:defW1loc}
        \|\Delta\|_{k\textnormal{-}W_1\textnormal{loc}} = \max\left\{\mathrm{Tr}\left[\Delta\,H\right] : H\in\mathcal{O}_{[n]}^{(k)}\,,\|H\|_{\textnormal{loc}}\le1\right\}\,.
    \end{equation}
The $k$-local quantum $W_1$ distance is the distance on $\mathcal{S}_{[n]}$ induced by the $k$-local quantum $W_1$ norm: For any 
and for any $\rho,\,\sigma\in\mathcal{S}_{[n]}$
    \begin{equation}\label{eq:dual def}
        \|\rho-\sigma\|_{k\textnormal{-}W_1\textnormal{loc}} = \max\left\{\mathrm{Tr}\left[\left(\rho-\sigma\right)H\right] : H\in\mathcal{O}_{[n]}^{(k)}\,,\|H\|_{\textnormal{loc}}\le1\right\}\,.
    \end{equation}
\end{defn}

\sp

\begin{rem}
Ref. \cite{depalma2023classical} defines a local quantum $W_1$ norm as the dual of the local norm \eqref{eq:locc}:
    \begin{equation}\label{eq:defW1locc}
        \|\Delta\|_{W_1\textnormal{loc}} = \max\left\{\mathrm{Tr}\left[\Delta\,H\right] : H\in\mathcal{O}_{[n]}^{(k)}\,,\|H\|_{\widetilde{\textnormal{loc}}}\le1\right\}\,.
    \end{equation}
The $k$-local quantum $W_1$ norm can be recovered from \eqref{eq:defW1locc} by setting to one the penalties associated with the regions containing at most $k$ spins and to infinity the penalties associated with the regions containing more than $k$ spins.
Indeed, with such choice the observables $H\in\mathcal{O}_{[n]}$ with $\|H\|_{\widetilde{\textnormal{loc}}}\le 1$ are exactly the $k$-local observables with $\|H\|_{\textnormal{loc}}\le1$, and the optimizations in \eqref{eq:defW1locc} and \eqref{eq:defW1loc} coincide.
\end{rem}

The $k$-local quantum $W_1$ norm can be computed with a linear program.
\eqref{eq:defW1loc} constitutes the dual program, while the primal program is provided by the following:

\sp

\begin{prop}\label{prop:primal}
    For any $\rho,\,\sigma\in\mathcal{S}_{[n]}$ we have
    \begin{equation}\label{eq:Wloc}
        \|\rho-\sigma\|_{k\textnormal{-}W_1\textnormal{loc}} = \min\left\{\sum_{x\in[n]}a_x : \left\|\rho_\Lambda-\sigma_\Lambda\right\|_1 \le 2\sum_{x\in\Lambda}a_x\;\forall\,\Lambda\subseteq[n]:|\Lambda|\le k\right\}\,.
    \end{equation}
\end{prop}

\begin{proof}
Analogous to \cite[Proposition 2.4]{depalma2023classical}.
\end{proof}

\begin{rem}
From \autoref{prop:primal}, the $k$-local quantum $W_1$ distance between $\rho$ and $\sigma$ depends only on the marginals $\rho_\Lambda$ and $\sigma_\Lambda$ on regions $\Lambda$ containing at most $k$ spins.
Therefore, if $\rho_\Lambda=\sigma_\Lambda$ for all such $\Lambda$, such distance vanishes even if $\rho$ and $\sigma$ do not coincide.
\end{rem}

\subsection{Properties of the \texorpdfstring{$k$}{k}-local quantum \texorpdfstring{$W_1$}{W1} distance}
The $k$-local quantum $W_1$ distance inherits all the properties of the local quantum $W_1$ distance of \cite{depalma2023classical}:
\begin{itemize}
\item For $n=1$, it coincides with the trace distance:
\begin{prop}\label{prop:trace}
    For $n=1$, $\|\cdot\|_{k\textnormal{-}W_1\textnormal{loc}} = \frac{1}{2}\|\cdot\|_1$.
\end{prop}
\begin{proof}
Follows from \autoref{prop:primal}.
\end{proof}
\item It is lower than the quantum $W_1$ distance of \cite{de2021quantum}:
\begin{prop}\label{prop:Lloc}
    We have $\|\cdot\|_{k\textnormal{-}W_1\textnormal{loc}} \le \|\cdot\|_{W_1}$.
\end{prop}

\begin{proof}
Analogous to \cite[Proposition 2.5]{depalma2023classical}.
\end{proof}

\item It can be upper bounded by the trace distances between the marginal of the states:
\begin{prop}\label{prop:W1locUB}
    For any $\rho,\,\sigma\in\mathcal{S}_{[n]}$ we have
    \begin{equation}\label{eq:DeltaR}
        \|\rho - \sigma\|_{k\textnormal{-}W_1\textnormal{loc}} \le \sum_{x\in[n]}\max_{\Lambda\ni x:\left|\Lambda\right|\le k}\frac{\left\|\rho_\Lambda - \sigma_\Lambda\right\|_1}{2\left|\Lambda\right|} \le n\max_{\Lambda\subseteq[n]:\left|\Lambda\right|\le k}\frac{\left\|\rho_\Lambda - \sigma_\Lambda\right\|_1}{2\left|\Lambda\right|}\,.
    \end{equation}
\end{prop}

\begin{proof}
Analogous to \cite[Proposition 2.6]{depalma2023classical}.
\end{proof}

\item It is superadditive with respect to the composition of quantum systems and additive with respect to the tensor product:
\begin{prop}\label{prop:tens}
    For any $\rho,\,\sigma\in\mathcal{S}_{[n]}$ and any $\Lambda\subset[n]$ we have
    \begin{equation}
        \|\rho-\sigma\|_{k\textnormal{-}W_1\textnormal{loc}} \ge \|\rho_\Lambda-\sigma_\Lambda\|_{k\textnormal{-}W_1\textnormal{loc}} + \|\rho_{\Lambda^c}-\sigma_{\Lambda^c}\|_{k\textnormal{-}W_1\textnormal{loc}}\,.
    \end{equation}
    Moreover, equality is achieved when $\rho = \rho_\Lambda\otimes\rho_{\Lambda^c}$ and $\sigma = \sigma_\Lambda\otimes\sigma_{\Lambda^c}$.
\end{prop}

\begin{proof}
Analogous to \cite[Proposition 2.7]{depalma2023classical}.
\end{proof}
\item It is lower bounded by the sum of the trace distances between the single-spin marginals:
\begin{prop}\label{prop:single}
For any $\rho,\,\sigma\in\mathcal{S}_{[n]}$ we have
\begin{equation}
    \|\rho-\sigma\|_{k\textnormal{-}W_1\textnormal{loc}} \ge \frac{1}{2}\sum_{x\in[n]}\|\rho_x - \sigma_x\|_1\,.
\end{equation}
\end{prop}
\begin{proof}
    Applying repeatedly \autoref{prop:tens} we get
    \begin{equation}
        \|\rho-\sigma\|_{k\textnormal{-}W_1\textnormal{loc}} \ge \sum_{x\in[n]}\|\rho_x - \sigma_x\|_{k\textnormal{-}W_1\textnormal{loc}} = \frac{1}{2}\sum_{x\in[n]}\|\rho_x - \sigma_x\|_1\,,
    \end{equation}
where the last equality follows from \autoref{prop:trace}.
The claim follows.
\end{proof}
\item It recovers the Hamming distance for the states of the computational basis:
\begin{cor}
    For any $x,\,y\in[q]^n$,
    \begin{equation}
    \left\||x\right\rangle\left\langle x| - |y\right\rangle\left\langle y|\right\|_{k\textnormal{-}W_1\textnormal{loc}} = h(x,y) = \left|\left\{i\in[n] : x_i\neq y_i\right\}\right|\,.
    \end{equation}
\end{cor}

\begin{proof}
Analogous to \cite[Corollary 2.1]{depalma2023classical}.
\end{proof}

\item It is contracting with respect to the action of single-spin quantum channels:
\begin{prop}
    For any $\Delta\in\mathcal{O}_{n}^T$ and any quantum channel $\Phi$ acting on a single spin, we have
    \begin{equation}
        \|\Phi(\Delta)\|_{k\textnormal{-}W_1\textnormal{loc}} \le \|\Delta\|_{k\textnormal{-}W_1\textnormal{loc}}.
    \end{equation}
\end{prop}

\begin{proof}
Analogous to \cite[Proposition 2.8]{depalma2023classical}.
\end{proof}
\end{itemize}

\subsection{A transportation-cost inequality}

The concentration inequality \autoref{thm:bound2} can be applied to prove that any quantum state $\sigma$ with exponentially decreasing correlations, as in \autoref{ass:length}, satisfies a transportation-cost inequality of the following form: 
\begin{equation}
f(\|\rho-\sigma\|_{k\textnormal{-}W_1\textnormal{loc}})\leq S_M(\rho\|\sigma)\qquad \forall \rho\in\mathcal S_{[n]}
\end{equation}
where $f$ is an increasing function and $S_M$ is the measured relative entropy.

\sp

\begin{defn}
    Let $\rho$ and $\sigma$ be density matrices on a finite-dimensional Hilbert space $\mathcal H$.
    The \emph{measured relative entropy} between $\rho$ and $\sigma$ is defined as:
\begin{equation}
    S_M(\rho\|\sigma):=\sup_{\{A_a\}_{a\in X}}\sum_{a\in X}\tr(A_a \rho)\log\frac{\tr(A_a\rho)}{\tr(A_a\sigma)},
\end{equation}
where the supremum is taken over all the POVMs $\{A_a\}_{a\in X}$ on $\mathcal H$.
\end{defn}

\sp

Let us recall that a POVM $\{A_a\}_{a\in X}$ on $\mathcal H$, where $X$ is a finite set, is a collection of positive semidefinite operators such that $\sum_{a\in X} A_a=I$ and provides the general mathematical description of a quantum measurement process.
Intuitively, the notion of measured relative entropy is given by the maximization, over all the measurement processes, of the Kullback-Leibler divergence between the classical probability distributions of the measurement outcomes associated to the considered states. The measured relative entropy admits the following characterization \cite{berta2017variational}:  

\sp

\begin{lem}\label{lem:entropy}
    The following idenity holds:
    \begin{equation}\label{eq:m_entropy}
        S_M(\rho\|\sigma)=\sup_{A>0}\left( \tr(\rho\log A)-\log\tr(\sigma A)\right).
    \end{equation}
\end{lem}

The variational expression \eqref{eq:m_entropy} can be used to bound from below the measured relative entropy by means of a function of the $k$-local quantum $W_1$ distance obtaining a transportation cost inequality.  

\sp

\begin{thm}\label{thm:bound_m_entropy}
    Let $\rho$ and $\sigma$ be states of $n$ quantum spins.
    Let $\sigma$ have correlation length $\xi$ as in \autoref{ass:length} with respect to a distance satisfying \autoref{ass:ball}.
    Then, the measured relative entropy between $\rho$ and $\sigma$ is bounded as follows: 
    \begin{equation}\label{eq:bound_m_entropy}
         S_M(\rho\|\sigma)\geq n^\frac{1}{2d+1} f(w)-\log(C+1),
    \end{equation}
    where $c$ is a positive constant, $w=\frac{1}{n}\|\rho-\sigma\|_{k\textnormal{-}W_1\textnormal{loc}}$, and
    \begin{equation}\label{eq:deff}
f(w)=\frac{F\left({k\,w}\right)}{k^2A^2\|H\|_{\textnormal{loc}}\left[\frac{\xi}{kA^2\|H\|_{\textnormal{loc}}}\left(ks^*(k\,w)+ F\left({k\,w}\right)\right)  \right]^{\frac{2d}{2d+1}} },
\end{equation}
where $F(x):=\max_ss(x-e^s+1+\frac{s}{2})$ and $s^*(x)=\textnormal{argmax}_s s(x-e^s+1+\frac{s}{2})$. 
\end{thm}

\begin{proof}
    We apply the claim of \autoref{lem:entropy} for $E=e^H$ where $H=\sum_{\Lambda\subseteq[n]} h_\Lambda$ is a self-adjoint operator given by the sum of local terms $h_\Lambda$ acting non trivially on $\Lambda\subseteq[n]$ such that $\left\langle h_\Lambda\right\rangle=0$ and  $|\Lambda|\leq k$  $\forall\Lambda$, where the angle brackets denote the expectation with respect to $\sigma$. Therefore:
    \begin{equation}\label{eq:D_M1}
        S_M(\rho\|\sigma)\geq \tr(\rho H)-\log \left\langle e^H\right\rangle\,.
    \end{equation}
    We can choose the Hamiltonian $H$ such that $\tr[(\rho-\sigma)H]=\|H\|_{\textnormal{loc}} \|\rho-\sigma\|_{k\textnormal{-}W_1\textnormal{loc}}$ using the fact that $\|\,\,\,\|_{k\textnormal{-}W_1 \textnormal{loc}}$ is the dual of $\|\,\,\,\|_{\textnormal{loc}}$ as provided by \eqref{eq:dual def}. Such a Hamiltonian can be always chosen in such a way that $\left\langle\sigma\right\rangle=0$, since $\|H\|_{\textnormal{loc}}=\|H-\left\langle\sigma\right\rangle I\|_{\textnormal{loc}}$ as $\|\,\,\,\|_{\textnormal{loc}}$ is a semi-norm with $\|I\|_{\textnormal{loc}}=0$. Moreover, by the arbitrary choice of the Hamiltonian, we can control the value of $\|H\|_{\textnormal{loc}}$.
    From \eqref{eq:D_M1}, we have:
    \begin{equation}
        S_M(\rho\|\sigma)\geq\tr[(\rho-\sigma)H]-\log \left\langle e^H\right\rangle\geq \|H\|_{\textnormal{loc}} \|\rho-\sigma\|_{k\textnormal - W_1 \textnormal{loc}}-\log \left\langle e^H\right\rangle,
    \end{equation}
Let us set $\|H\|_{\textnormal{loc}}=t$ and $\|\rho-\sigma\|_{k\textnormal{-}W_1\textnormal{loc}}=W$, then: 
\begin{equation}\label{eq:D_M2}
    S_M(\rho\|\sigma)\geq tW -\log \left\langle e^H\right\rangle.
\end{equation}
Taking the exponential of \eqref{eq:D_M2} we obtain:
\begin{equation}
    e^{-S_M(\rho\|\sigma)}\leq e^{-tW}\left\langle e^H\right\rangle=e^{-tnw}\left\langle e^H\right\rangle.
\end{equation}
This bound can be improved optimizing w.r.t. $t$ following the argument used in the proof of \autoref{thm:concentration2}. More precisely, we can apply \autoref{thm:bound2} to bound $\left\langle e^H\right\rangle$ taking into account that $\left\langle H\right\rangle=0$, thus:
\begin{equation}
  e^{-S_M(\rho\|\sigma)} \leq \exp\left[-\frac{n}{k^2A^2 l^{2d}{\|H\|_{\textnormal{loc}}}} s\left({k\,w}-e^s+1+\frac{s}{2}\right)  \right]+ C\exp\left[\frac{ns}{k A^2\, l^{2d}}-\frac{l}{\xi}  \right],
\end{equation}
where $s=tkA^2l^{2d}\|H\|_{\textnormal{loc}}$. As in the proof of \autoref{thm:concentration2}, we can improve the bound minimizing over $s$ and requiring the balance of the two exponential terms to fix $l$, this yields: 
\begin{equation}
e^{-S_M(\rho\|\sigma)}\leq (C+1)\exp\left[-\frac{n^{\frac{1}{2d+1}}F\left({k\,w}\right)}{k^2A^2\|H\|_{\textnormal{loc}}\left[\frac{\xi}{kA^2\|H\|_{\textnormal{loc}}}\left(k\,s^*(k\,w)+ F\left({k\,w}\right)\right)  \right]^{\frac{2d}{2d+1}} } \right]\,.
\end{equation}
We then obtain the inequality
\begin{equation}
\exp\left[{-S_M(\rho\|\sigma)}\right] \leq \exp\left[\log(C+1)-n^{\frac{1}{2d+1}}f(w)\right]\,,
\end{equation}
which implies the claim.

\end{proof}

\section{Equivalence of the ensembles of quantum statistical mechanics}\label{sec:ensembles}

The three main ensembles employed in quantum statistical mechanics to compute the equilibrium properties of quantum systems are the canonical ensemble, the microcanonical ensemble and the diagonal ensemble.
The quantum state associated to the canonical ensemble is the Gibbs state, which describes the physics of a system that is at thermal equilibrium with a large bath at a given temperature.
The diagonal and microcanonical ensembles both describe the physics of an isolated quantum system, and the associated states are convex combinations of the eigenstates of the Hamiltonian.
The microcanonical ensemble assumes a uniform probability distribution for the energy in a given energy shell.
The diagonal ensemble includes all the states that are diagonal in the eigenbasis of the Hamiltonian, and in particular it includes the eigenstates themselves.
For many quantum systems, the canonical and microcanonical ensembles give the same expectation values for local observables if the corresponding states have approximately the same average energy.
A lot of effort has been devoted to determining conditions under which the two ensembles are equivalent \cite{lima1971equivalence,lima1972equivalence,muller2015thermalization,brandao_equivalence_2015}.
The most prominent among such conditions are short-ranged interactions and a finite correlation length, but so far analytical proofs could be obtained only in the case of regular lattices \cite{brandao_equivalence_2015}.
The situation is more complex considering also the diagonal ensemble which includes all the states that are diagonal in the eigenbasis of the Hamiltonian.
The condition under which this ensemble is equivalent to the microcanonical and canonical ensembles is called Eigenstate Thermalization Hypothesis (ETH) \cite{deutsch1991quantum,srednicki_chaos_1994,gogolin2016equilibration,reimann2018dynamical,de2015necessity}, and states that the expectation values of local observables on the eigenstates of the Hamiltonian are a smooth function of the energy, \emph{i.e.}, for any given local observable, \emph{any} two eigenstates with approximately the same energy yield approximately the same expectation value.
The ETH is an extremely strong condition on the Hamiltonian and several quantum systems, including all integrable systems, do not satisfy it.
A weak version of the ETH has been formulated \cite{biroli2010effect,kuwahara2020eigenstate}, stating that for any given local observable, \emph{most} eigenstates in an energy shell yield approximately the same expectation value, or, more precisely, that the fraction of eigenstates yielding expectation values far from the Gibbs state with the same average energy vanishes in the thermodynamical limit.
The weak ETH implies the equivalence between the canonical and microcanonical ensembles, but is not sufficient to prove their equivalence with the diagonal ensemble.
Under the hypothesis of finite correlation length in the Gibbs state, an analytical proof of the weak ETH is available only for regular lattices \cite{kuwahara2020eigenstate}.

In this paper, we consider the equivalence of the statistical mechanical ensembles and the weak ETH from the perspective of the $k$-local quantum $W_1$ distance.
The closeness in such distance implies closeness of the expectation values of all the $k$-local observables with $O(1)$ local norm.
We stress that such observables do not need to be geometrically local, since they can contain terms acting on spins at arbitrary distance.
Our results are based on the transportation-cost inequality \autoref{thm:bound_m_entropy} and are similar in spirit to the results of \cite[Sec. 8]{de2022quantum}, which considers the equivalence of the ensembles with respect to the quantum $W_1$ distance of \cite{de2021quantum} employing transportation-cost inequalities for such distance.
The key difference between the present paper and Ref. \cite{de2022quantum} is that the $k$-local quantum $W_1$ distance that we employ is weaker than the quantum $W_1$ distance of \cite{de2022quantum}.
As a consequence, our \autoref{thm:bound_m_entropy} requires only exponentially decaying correlations.
On the contrary, the transportation-cost inequalities employed in \cite{de2022quantum} require a Hamiltonian that is the sum of local commuting terms (the subsequent Ref. \cite{onorati2023efficient} has proved transportation-cost inequalities for the quantum $W_1$ distance under the hypothesis of exponential decay of the conditional mutual information, which is however still stronger than exponential decay of correlations).
Therefore, we are able to prove the equivalence of the ensembles for a significantly larger class of Hamiltonians than the class for which the results of \cite[Sec. 8]{de2022quantum} apply, at the price of a weaker notion of equivalence.
We also stress that, contrarily to the results of Refs. \cite{brandao_equivalence_2015,kuwahara2020eigenstate}, we do not need to restrict to regular lattices, and our results do not require any regularity nor the notion of neighboring spins, but only the upper bound \autoref{ass:ball} to the growth of the size of the metric balls with the distance.

Let $H\in\mathcal{O}_{[n]}$ be the Hamiltonian of a quantum system of $n$ spins. A \emph{Gibbs state} for the Hamiltonian $H$ is:
\begin{equation}\label{eq:Gibbs}
    \omega:=\frac{e^{-\beta H}}{\tr(e^{-\beta H})},
\end{equation}
where $\beta$ is the inverse temperature.
The following \autoref{prop:equiv} implies that any state $\rho\in\mathcal{S}_{[n]}$ is close in the $k$-local $W_1$ distance to the Gibbs state $\omega$ with the same average energy, provided that $\omega$ has finite correlation length and approximately the same entropy as $\rho$, \emph{i.e.},
\begin{equation}
    S(\omega) - S(\rho) = o\left(n^\frac{1}{2d+1}\right)\,.
\end{equation}
Moreover, under the same hypothesis, the average reduced states over one spin of $\rho$ and $\omega$ are close in trace distance.

\sp

\begin{prop}\label{prop:equiv}
Let $\omega\in\mathcal{S}_{[n]}$ be a Gibbs state for the Hamiltonian $H\in\mathcal{O}_{[n]}$ with correlation length $\xi$ as in \autoref{ass:length} with respect to a distance satisfying \autoref{ass:ball}.
Then, any quantum state $\rho\in\mathcal{S}_{[n]}$ with the same average energy as $\omega$ satisfies
\begin{equation}
    \frac{1}{n}\left\|\rho-\omega\right\|_{k\textnormal{-}W_1\textnormal{loc}} \le f^{-1}\left(\frac{S(\omega) - S(\rho)}{n^\frac{1}{2d+1}}+c\right)\,,
\end{equation}
with $f$ as in \eqref{eq:deff}.
Moreover, let $\Lambda:\mathcal{S}_{[n]}\to\mathcal{S}(\mathbb{C}^d)$ be the quantum channel that computes the average marginal state over one spin, \emph{i.e.}, for any $\rho\in\mathcal{S}_{[n]}$,
\begin{equation}
    \Lambda(\rho) = \frac{1}{n}\sum_{x\in [n]}\rho_x\,.
\end{equation}
Then,
\begin{equation}
    \frac{1}{2}\left\|\Lambda(\rho) - \Lambda(\omega)\right\|_1 \le f^{-1}\left(\frac{S(\omega) - S(\rho)}{n^\frac{1}{2d+1}}+c\right)\,.
\end{equation}
\end{prop}
\begin{proof}
We have
\begin{equation}
    S_M(\rho\|\omega) \le S(\rho\|\omega) = S(\omega) - S(\rho)\,,
\end{equation}
where the last equality follows since $\rho$ and $\omega$ have the same average energy.
We then have from \autoref{thm:bound_m_entropy}
\begin{equation}
    \frac{1}{n}\left\|\rho-\omega\right\|_{k\textnormal{-}W_1\textnormal{loc}} \le f^{-1}\left(\frac{S_M(\rho\|\omega)}{n^\frac{1}{2d+1}}+c\right) \le f^{-1}\left(\frac{S(\omega) - S(\rho)}{n^\frac{1}{2d+1}}+c\right)\,.
\end{equation}

We have from \autoref{prop:single} and \autoref{prop:equiv}
\begin{align}
    \left\|\Lambda(\rho) - \Lambda(\omega)\right\|_1 \le \frac{1}{n}\sum_{x\in[n]}\left\|\rho_x - \omega_x\right\|_1 \le \frac{2}{n}\left\|\rho - \omega\right\|_{k\textnormal{-}W_1\textnormal{loc}} \le 2\,f^{-1}\left(\frac{S(\omega) - S(\rho)}{n^\frac{1}{2d+1}}+c\right)\,,
\end{align}
and the claim follows.
\end{proof}

Choosing $\rho$ to be diagonal in the eigenbasis of the Hamiltonian, \autoref{prop:equiv} implies that any convex combination of a sufficiently large number of eigenstates is close in $k$-local $W_1$ distance to the Gibbs state with the same average energy.
Such number of eigenstates can be a fraction $\exp\left(-o\left(n^\frac{1}{2d+1}\right)\right)$ of the total number of eigenstates appearing in a microcanonical state.
Therefore, \autoref{prop:equiv} constitutes an exponential improvement over the weak ETH.

\subsection{Width of the microcanonical energy shell}
We can apply the transportation-cost inequality provided by \autoref{thm:bound_m_entropy} to prove the convergence of the microcanonical ensemble to the canonical ensemble in the local $W_1$ distance as $n\rightarrow +\infty$.
Let $H=\sum_E E P(E)$ be the spectral decomposition of $H$. The \emph{microcanonical state} with energy $E$ and width $\Delta$ is 
\begin{equation}\label{eq:microcanonical}
    \omega_{E,\Delta}:=\frac{P(E,\Delta)}{\tr(P(E,\Delta))},
\end{equation}
where $P(E,\Delta)$ is the orthogonal projector onto the subspace spanned by the eigenvectors of $H$ with eigenvalues in $(E-\Delta, E]$.

\sp

\begin{thm}\label{thm:equiv}
    Let $\omega$ be a Gibbs state for the Hamiltonian $H\in\mathcal{O}_{[n]}$ with correlation length $\xi$ as in \autoref{ass:length} with respect to a distance satisfying \autoref{ass:ball} and let $\omega_{E^*,\Delta}$ be the microcanonical state with energy
    \begin{equation}
        E^*=\mathrm{argmax}_E\, e^{-\beta E} \tr\,P(E,\Delta)\,.
    \end{equation}
    Then:
    \begin{equation}
        \frac{1}{n}\|\omega_{E^*,\Delta}-\omega\|_{k\textnormal{-}W_1 loc}=o_{n\rightarrow \infty}(1) ,
    \end{equation}
    whenever $\Delta= \exp\left(-o\left(n^{\frac{1}{2d+1}}\right)\right)$ and $\Delta=o\left(n^{\frac{1}{2d+1}}\right)$.
\end{thm}

\begin{proof}
    \autoref{thm:bound_m_entropy} provides a bound on the measured relative entropy as a function of the $k$-local $W_1$ distance. Since $S(\omega_{E,\Delta}\|\omega)\geq S_M(\omega_{E,\Delta}\|\omega)$ let us control the relative entropy between the Gibbs and the microcanonical states:
    \begin{equation}\label{eq:boundS}
        S(\omega_{E^*,\Delta}\|\omega)=\ln\frac{\tr[e^{-\beta H}]}{\tr[P(E^*,\Delta)]}+\beta\ln\left[H\frac{P(E^*,\Delta)}{\tr[P(E^*,\Delta)]}\right]\leq \ln\frac{\tr[e^{-\beta H}]}{\tr[P(E^*,\Delta)]}+\beta E^*.
    \end{equation}
Let us observe that:
\begin{equation}
    \tr\left[\sum_{E} \frac{e^{-\beta E}}{\tr[e^{-\beta H}]}P(E) \right]=1,
\end{equation}
where the sum over $E$ is taken over all the spectrum of $H$.
Let us consider the partition function:
\begin{equation}\label{eq:Z}
    Z:=\tr\left[\sum_{E}e^{-\beta E} P(E)\right]=  \tr[e^{-\beta H}],
\end{equation}
and define the interval $I=(-\|H\|-\Delta, \|H\|]$, then:
\begin{equation}
    Z\leq \sum_{\nu\in\mathbb Z \,:\, \nu\Delta\in I} \tr\left[P(\nu\Delta,\Delta)\right] e^{-\beta \Delta(\nu-1).}
\end{equation}
By definition of $E^*$, we have:
\begin{equation}
    \tr[P(\nu\Delta,\Delta)] e^{-\beta\Delta(\nu-1)}\leq e^{\beta(\Delta-E^*)} \tr[P(E^*,\Delta)],
\end{equation}
therefore:
\begin{equation}\label{eq:boundZ}
    Z\leq \frac{2\|H\|+\Delta}{\Delta} e^{\beta(\Delta-E^*)} \tr[P(E^*,\Delta)].
\end{equation}
From equations \eqref{eq:Z} and \eqref{eq:boundZ}, we obtain:
\begin{equation}
   \frac{\tr[e^{-\beta H}]}{\tr[P(E^*,\Delta)]}\leq \left(\frac{2\|H\|}{\Delta}+1\right) e^{\beta(\Delta-E^*)},
\end{equation}
that is:
\begin{equation}\label{eq:log}
    \ln\left[\frac{\tr[e^{-\beta H}]}{\tr[P(E^*,\Delta)]}\right]\leq \ln\left[ \frac{2\|H\|}{\Delta} +1  \right]+\beta(\Delta-E^*).
\end{equation}
Without loss of generality, we can add a constant to $H$ such that its minimum and maximum eigenvalues become opposite.
Then, $\|H\|\leq \frac{n}{2}\|H\|_\textnormal{loc}$ by \autoref{lem:norms}. Inserting \eqref{eq:log} into \eqref{eq:boundS}, we have:
\begin{equation}
    S(\omega_{E^*,\Delta}\|\omega)\leq \ln\left[ \frac{n\|H\|_\textnormal{loc}}{\Delta}+1\right]+\beta\Delta.
\end{equation}
Applying \autoref{thm:bound_m_entropy} to bound $S(\omega_{E^*,\Delta}\|\omega)$ from below, we obtain:
    \begin{equation}\label{eq:delta,w}
    \ln\left[ \frac{n \|H\|_\textnormal{loc}}{\Delta}+1\right]+\beta\Delta+c\geq n^{\frac{1}{2d+1}}f(w),    
    \end{equation}
    where $c>0$ and $f$ is the positive increasing function introduced in \eqref{eq:bound_m_entropy}. 
   Therefore:
    \begin{equation}
        w\leq f^{-1}\left( n^{-\frac{1}{2d+1}} \ln\left[ \frac{n \|H\|_\textnormal{loc}}{\Delta}+1\right]+  n^{-\frac{1}{2d+1}}\beta\Delta+  n^{-\frac{1}{2d+1}}c       \right)\,.
    \end{equation}
    The right-hand side of the inequality above is $o(1)$ whenever $\ln(n/\Delta)=o\left(n^{\frac{1}{2d+1}}\right)$, that is $\Delta=n \exp\left(-o\left(n^{\frac{1}{2d+1}}\right)\right) = \exp\left(-o\left(n^{\frac{1}{2d+1}}\right)\right)$, and $\Delta=o\left(n^{\frac{1}{2d+1}}\right)$.
The claim follows.
\end{proof}

A similar result on the equivalence between the quantum canonical and microcanonical ensembles under the hypothesis of exponentially decaying correlations is proved in \cite{kuwahara2020eigenstate}. Such results states that for any local observable $O=\frac{1}{n}\sum_{v\in[n]} O_v$, where each $O_v$ is supported on the ball centered in $v$ with fixed radius and $\|O_v\|\leq 1$, the expectations over the Gibbs and the microcanonical states converges to the same value as $n\rightarrow +\infty$ whenever $\Delta=\exp\left(-O\left(n^\frac{1}{d+1}\right)\right)$. While our result poses a stronger requirement on the width of the energy shell, we do not have any requirement on the radius of the regions where the local terms of $O$ are supported, since we require such regions only to have bounded cardinality.

\section{Conclusions}\label{sec:concl}

In this paper we have proved new concentration inequalities for quantum spin systems and we have applied such inequalities to prove the equivalence between the canonical and microcanonical ensembles of quantum statistical mechanics.

First, we have introduced the \emph{local norm} \eqref{eq:loc} for observables of quantum spin systems.
Then, we have proved a concentration inequality for local observables whose local norm is $O(1)$ measured on product states (\autoref{thm:concentrationproduct}), stating that the probability of deviations from the average decays exponentially with the number of spins.
We have extended the result to the states with finite correlation length as in \autoref{ass:length}, proving that the probability of deviations from the average decays quasi exponentially with the number of spins (\autoref{thm:concentration2}). Remarkably, our results do not require any regular structure nor any notion of neighboring spins.

Moreover, we have defined the $k$-local quantum $W_1$ distance as the distance on the states of a quantum spin system that quantifies the distinguishability with respect to $k$-local observables.
We have proved a transportation-cost inequality stating that the $k$-local quantum $W_1$ distance between a generic state and a state with exponentially decaying correlations is upper bounded by a function of the relative entropy between the states (\autoref{thm:bound_m_entropy}).
Such inequality implies that any Gibbs state $\omega$ with exponentially decaying correlations is close w.r.t. the $k$-local $W_1$ distance, to any state $\rho$ with the same average energy and such that $S(\omega)-S(\rho)=o\left(n^{\frac{1}{2d+1}}\right)$ (\autoref{prop:equiv}), where $d$ is the power of the radius governing the growth of the volume of the metric
balls of the set of the spins.
This result implies the equivalence between the canonical and the microcanonical ensembles of quantum statistical mechanics for all the Hamiltonians whose Gibbs states have exponentially decaying correlations and all the microcanonical states whose energy shell is wide enough.
Furthermore, \autoref{prop:equiv} implies that any convex combination of eigenstates which constitute a fraction $\exp\left(-o\left(n^{\frac{1}{2d+1}}\right)\right)$ of the total number of eigenstates in a microcanonical shell is close in the $k$-local $W_1$ distance to the Gibbs state. This is an exponential improvement over the weak ETH.
Our results do not require any regular structure nor any notion of neighboring spins, and hold at any temperature provided the correlations in the Gibbs state decay exponentially with the distance.   

The main open problem we point out concerns the strong ETH, stating that the energy eigenstate themselves are close to the Gibbs states with the same average energy with respect to a suitable class of observables. From \cite[Sec. 8]{de2022quantum}, the strong ETH cannot be captured by the quantum $W_1$ distance of \cite{de2021quantum}.
Indeed, the von Neumann entropy per spin is continuous with respect to the $W_1$ distance per spin \cite{de2023wasserstein}, therefore no pure state can be close to any state with large entropy such a Gibbs state.
On the contrary, we do not expect such entropy continuity to hold for the $k$-local quantum $W_1$ distance, and therefore there is hope that this distance can capture the strong ETH.

\section*{Acknowledgements}
This work was partially supported by project SERICS (PE00000014) under the MUR National Recovery and Resilience Plan funded by the European Union - NextGenerationEU.

GDP has been supported by the HPC Italian National Centre for HPC, Big Data and Quantum Computing - Proposal code CN00000013 - CUP J33C22001170001 and by the Italian Extended Partnership PE01 - FAIR Future Artificial Intelligence Research - Proposal code PE00000013 - CUP J33C22002830006 under the MUR National Recovery and Resilience Plan funded by the European Union - NextGenerationEU.
Funded by the European Union - NextGenerationEU under the National Recovery and Resilience Plan (PNRR) - Mission 4 Education and research - Component 2 From research to business - Investment 1.1 Notice Prin 2022 - DD N. 104 del 2/2/2022, from title ``understanding the LEarning process of QUantum Neural networks (LeQun)'', proposal code 2022WHZ5XH – CUP J53D23003890006.
GDP and DP are members of the ``Gruppo Nazionale per la Fisica Matematica (GNFM)'' of the ``Istituto Nazionale di Alta Matematica ``Francesco Severi'' (INdAM)''.

\section*{Declarations}

%Some journals require declarations to be submitted in a standardised format. Please check the Instructions for Authors of the journal to which you are submitting to see if you need to complete this section. If yes, your manuscript must contain the following sections under the heading `Declarations':

\begin{itemize}
%\item Funding Not 
\item Conflict of interest: All authors have no conflicts of interest.
%\item Ethics approval and consent to participate
%\item Consent for publication
\item Data availability: The authors declare that the present manuscript has no associated data.
\end{itemize}

%\noindent
%If any of the sections are not relevant to your manuscript, please include the heading and write `Not applicable' for that section. 

\begin{appendices}

\section{Auxiliary lemmas}\label{sec:aux}
\begin{lem}\label{lem:norms}
Let $H\in\mathcal{O}_{[n]}$ such that its minimum and maximum eigenvalues are opposite.
Then, $2\,\|H\| \le n\,\|H\|_{\textnormal{loc}}$.
\end{lem}

\begin{proof}
    Let
    \begin{equation}
H = \sum_{\Lambda\subseteq[n]}h_\Lambda\,,\qquad h_\Lambda\in\mathcal{O}_\Lambda
    \end{equation}
such that
\begin{equation}\label{eq:Hdec}
    \|H\|_{\textnormal{loc}} = 2\max_{x\in[n]}\sum_{\Lambda\ni x}\|h_\Lambda\|\,.
\end{equation}
We have
\begin{align}
    \|H\| &\le \|H-h_\emptyset\| \le \sum_{\Lambda\subseteq[n] : \Lambda\neq\emptyset}\|h_\Lambda\| \le \sum_{\Lambda\subseteq[n]}|\Lambda|\,\|h_\Lambda\| = \sum_{\Lambda\subseteq[n]}\sum_{x\in\Lambda}\|h_\Lambda\|\nonumber\\
    &= \sum_{x=1}^n\sum_{\Lambda\ni x}\|h_\Lambda\| \le \frac{n}{2}\,\|H\|_{\textnormal{loc}}\,.
\end{align}
The claim follows.
\end{proof}

\begin{lem}\label{lem:loc}
     Let $H\in\mathcal{O}_{[n]}$ and let $H = \sum_{\Lambda\subseteq[n]}\tilde{h}_\Lambda$
     be the decomposition of $H$ that achieves the local norm as in \eqref{eq:loc}.
     For any $\emptyset\neq\Lambda\subseteq[n]$, let $h_\Lambda := \tilde{h}_\Lambda - \left\langle\tilde{h}_\Lambda\right\rangle$, where the angle brackets denote the expectation with respect to a generic state.
     Then,
     \begin{equation}
         \sum_{\Lambda \ni v} \|h_\Lambda\|\leq \|H\|_{\textnormal loc}\qquad \forall\, v\in[n].
     \end{equation}
\end{lem}
\begin{proof}
    By definition of $h_\Lambda$, we have $\|h_\Lambda\|\leq \|\tilde h_\Lambda\|+\left|\left\langle \tilde h_\Lambda\right\rangle\right|$ for every $\emptyset\not =\Lambda\subseteq[n]$ then 
\begin{equation}\label{eq:ineqloc}
\sum_{\Lambda \ni v} \|h_\Lambda\|\leq \sum_{\Lambda \ni v} \|\tilde h_\Lambda\|+\sum_{\Lambda \ni v} \left|\left\langle\tilde h_\Lambda\right\rangle\right|\leq 2\sum_{\Lambda \ni v} \|\tilde h_\Lambda\|\leq \|H\|_{\textnormal loc}\qquad \forall v\in[n],    
\end{equation}
where we have respectively used: the triangle inequality, the operator norm inequality $\|\tilde h_\Lambda\|\geq \left\langle \tilde h_\Lambda\right\rangle$, and $\|H\|_{\textnormal{loc}}=2\max_{v\in[n]} \sum_{\Lambda\ni v}\|\tilde h_\Lambda\|$.
\end{proof}

\end{appendices}

%\printbibliography %Prints bibliography

\bibliography{lattices}

\end{document}